\documentclass{llncs}
\usepackage{makeidx}  
\usepackage{graphicx}
\usepackage{psfrag}
\usepackage{amsfonts}
\usepackage[all]{xy}
\newcommand{\Ke}{\kern0.4em}
\newcommand{\ke}{\kern0.2em}

\newcommand{\disp}{\displaystyle}

\newcommand{\from}{\colon}
\newcommand{\too}{\longrightarrow}
\newcommand{\mapstoo}{\longmapsto}

\newcommand{\q}{\quad}

\newcommand{\bb}[1]{\mathbb#1}

\newcommand{\sm}{\!\setminus\!}

\newcommand{\supp}{\mathop{\rm supp}\nolimits}

\newcommand{\End}{\mathop{\rm End}\nolimits}
\newcommand{\Hom}{\mathop{\rm Hom}\nolimits}

\newcommand{\Herm}{\mathop{\rm Herm}\nolimits}

\newcommand{\by}{\!\times\!}

\renewcommand{\phi}{\varphi}

\newtheorem{thm}{Theorem}
\newtheorem{prop}[thm]{Proposition}
\newtheorem{la}[thm]{Lemma}
\newtheorem{cor}[thm]{Corollary}

\begin{document}
\bibliographystyle{plain}
\frontmatter          
\pagestyle{headings}  
\addtocmark{Automorphisms of stabilizer codes} 
%
%
\mainmatter              
\title{Automorphisms of stabilizer codes}
\titlerunning{Automorphisms of stabilizer codes}  
%
\author{Klaus Wirthm\"uller}
\authorrunning{Klaus Wirthm\"uller}   
%
\tocauthor{Klaus Wirthm\"uller}
\institute{Fachbereich Mathematik der
           Technischen Universit\"at Kaiserslautern,\\
           Postfach 3049, 67653 Kaiserslautern, Germany\\
           \email{wirthm@mathematik.uni-kl.de}}

\maketitle              

\begin{abstract}
We study the automorphisms of binary stabilizer codes and states. We prove that they almost always form a solvable group, and thereby shed new light on the fact that there is no universal set of transversal gates. We also determine the connected component of the automorphism group.

\end{abstract}
\section{Introduction}
\label{int}
\noindent
A stabilizer code is an error-correcting quantum code that can be realised as the common fixed point set of mutually commuting products of Pauli operators. These codes, first described in \cite{gottesman_class_error_corr_codes}, are a widely used tool in fault-tolerant quantum computation. To take advantage of them it is essential to be able to implement quantum gates that act directly on encoded states, if possible by operators that act near-locally on each qubit of the code words. In the case of true locality such gates are called transversal \cite{gottesman_thesis}\ts; they always include a full set of Pauli operators on the encoded qubits. The desire to have many transversal gates points to the automorphism group of a stabilizer code as an object of interest. In fact it has been shown in \cite{transversality_vs_universality}, and in \cite{universal_not_even-one} for a more general setting, that there are not sufficiently many transversal gates in order to achieve universal computation\ts: for that purpose the automorphism group is never as large as one might wish.

From a different point of view stabilizer codes are quantum states rather than codes, mixed in general but including pure ones. While the natural classification of quantum states would be by local unitary equivalence, for stabilizer states the so-called Clifford equivalence is much more accessible since it takes the additional structure into account. The relation between these classifications has been intensively studied \cite{lu_vs_lCliff_stab,lu_vs_lCliff_stab_graph,lulc_diag_qforms}, and it has been eventually seen \cite{lu_lc_is_false} that they fail to coincide. The discrepancy too is related to the automorphism group --- ironically to the fact that it can be larger than one might expect.

The purpose of this note is twofold. Theorem \ref{pointwise_groups}, the main statement of Section \ref{small}, implies that the automorphism group of a stabilizer code is, with explicitly known exceptions, a solvable group. The fact that no universal set of transversal gates can exist, now becomes an immediate corollary and acquires a new conceptual meaning this way. Theorem \ref{pointwise_groups} itself occurs implicitly in  \cite{transversality_vs_universality}, and our proof follows closely the reasoning there, which in turn uses ingredients from \cite{quant_codes_min_dist_two,lu_vs_lCliff_stab,lu_vs_lCliff_stab_graph,lulc_diag_qforms}. We think it worthwhile to assemble it here as a whole since it takes advantage of several shortcuts, as in the proof of Rains' Lemma \cite{quant_codes_min_dist_two} which plays a basic role.

In Section \ref{connected_group} we determine, for a given stabilizer code, the connected component of the unit element of its automorphism group. It turns out that this connected automorphism group --- necessarily a compact torus --- is largely controlled by those elements of the stabilizer groups which act on pairs of qubits\ts: see Theorem \ref{connected} for the precise statement. It is then easy to recognise the states with a finite automorphism group as well as the codes which allow but a finite number of transversal gates.
\section{The automorphism group is solvable}
\label{small}
\noindent
We work in the setting of binary stabilizer states and codes over a finite set $P$ of qubits. We write the standard base of the binary vector space $V:=\bb F_2\oplus\bb F_2$ as $(z,x)$ and equip $V$ with its unique symplectic form $\omega$. The basic objects are vector subspaces $L\subset V^P$ which are isotropic\ts: thus $\omega$ vanishes identically on $L\by L$. Every such subspace $L\subset V^P$ of dimension $|P|\!-\!d$ determines a class of stabilizer codes in the state space ${\cal H}^{\otimes P}:=(\bb C^2)^{\otimes P}$, encoding $d$ qubits. Thus the Lagrangian case $d=0$ is that which correspond to (pure) stabilizer \textit{states}.

In order to fix a stabilizer code within its class let
\begin{displaymath}
  LU(P)\subset U({\cal H}^{\otimes P})
\end{displaymath}
denote the group of local unitary automorphisms\ts; it is the quotient group of the Cartesian product $U(2)^P$ by the co-diagonal of scalars ${(\lambda_s)}_{s\in P}$ with $\prod\lambda_s=1$. Recall that the Pauli group is the subgroup $\tilde V^P\subset LU(P)$ generated by all tensor products of Pauli operators and arbitrary scalars\ts; the commutators in this group are given by the values of the symplectic form $\omega$. The normaliser of $\tilde V^P$ in $LU(P)$ or --- more often --- its quotient by the subgroup of scalars is known as the Clifford group $\hbox{Clif}\,(P)$. It admits a canonical representation
\begin{displaymath}
  \xymatrix{{0}\ar@{->}[r] &{V^P\strut}\ar@{->}[r]
           &{\hbox{Clif}\,(P)}\ar@{->}[r]
           &{L\hbox{Sp}(P)}\ar@{->}[r]&{1}\\}
\end{displaymath}
as an extension of the local symplectic group of $V^P$ by the additive group $V^P$.

If $L\subset V^P$ is isotropic the natural projection $\tilde V^P\to V^P$ admits (homomorphic) sections over $L$. Any choice $\sigma\from L\to\tilde V^P$ of such a section singles out the \textit{stabilizer group} $\tilde L:=\sigma(L)\subset\tilde V^P$ and thereby determines a unique stabilizer code
\begin{displaymath}
  C(L,\sigma):=\left({\cal H}^{\otimes P}\right)^{\tilde L}
              \subset{\cal H}^{\otimes P}
\end{displaymath}
as its fixed point space. If a different section $\tau$ is chosen then the quotient $\sigma/\tau$ takes values in $\{\pm1\}$ and thus defines a homomorphism $\sigma/\tau\in\Hom(L,\bb F_2)$. In terms of the symplectic complement $L^\perp\subset V^P$ of $L$ the latter space is the same as $V^P/L^\perp$, and any element $h\in\tilde V^P$ that represents $\sigma/\tau$ there will send $C(L,\sigma)$ isomorphically onto $C(L,\tau)$. In particular all codes associated to $L$ are Clifford equivalent among themselves, and we will sometimes drop $\sigma$ from the notation when its choice is irrelevant to a statement.

We are concerned with the local unitary automorphisms of $C(L,\sigma)$, or all families $\tilde g={(\tilde g_s)}_{s\in P}\in U(2)^P$ that send $C(L,\sigma)$ into itself. As this group clearly contains the $|P|$-dimensional torus of scalars we may as well divide by this subgroup\ts; using the standard identification $U(2)/S^1=SO(3)$ via the adjoint action on traceless Hermitian operators we therefore define
\begin{displaymath}
  A(L,\sigma):=\left\{g\in SO(3)^P\,\big|\,
                      \tilde g\bigl(C(L,\sigma)\bigr)\subset C(L,\sigma)\hbox{ for some }\tilde g\hbox{ representing }g\right\}
\end{displaymath}
as the \textit{automorphism group} of the code. By definition it contains a copy of $L\simeq\tilde L$, with each Pauli operator acting as a rotation by 180 degrees around the corresponding axis.

It may happen that the isotropic subspace $L\subset V^P$ is decomposable in the sense that it is a direct product $L=L_S\times L_{P\setminus S}$ for some subset $S$ of qubits with $\emptyset\neq S\neq P$, where quite generally we identify $\{v\in V^P\,|\,v_p=0\hbox{ for all }p\in P\sm S\}$ with $V^S$ and put
\begin{displaymath}
  L_S:=L\cap V^S.
\end{displaymath}
Every $L$ thus has a unique representation as a direct product $L=\prod_SL_S$ where the subsets $S\subset P$ form a partition of $P$, and each $L_S$ is indecomposable. We then clearly have corresponding decompositions
\begin{displaymath}
  C(L,\sigma)=\bigotimes_SC(L_S,\sigma|L_S)\hbox{\q and\q}
  A(L,\sigma)=\prod_SA(L_S,\sigma|L_S),
\end{displaymath}
and therefore it is sufficient to study the automorphism groups of indecomposable isotropic subspaces $L\subset V^P$.

We first treat two particularly simple exceptional cases.
\begin{enumerate}
  \item In case $|P|=1$ we have either $L=0$ --- resulting in
    $A(L)=SO(3)$, or after a suitable symplectic automorphism $L$ is spanned
    by the single vector $z\in V$. We then may assume that
    $\tilde L=\langle Z\rangle$ and therefore $C(L,\sigma)$ is spanned by
    the stabilizer state $|0\rangle$. The automorphism group clearly is
    the circle group $A(L,\sigma)=C_z\simeq SO(2)$ of rotations around the
    $z$ axis. In this case we refer to the (non-zero) Lagrangian $L$, as
    well as the unique qubit of $P$ as \textit{trivial}.
\medskip
  \item The other exceptional case is that of an indecomposable Lagrangian
    with $|P|=2$. Normalising by a local symplectic automorphism we may
    assume $L=\langle zz, xx\rangle$, and the obvious choice of $\sigma$
    leads to the Bell state $|00\rangle+|11\rangle$. In terms of the
    isomorphy ${\cal H}\otimes{\cal H}=\End{\cal H}$ it corresponds to the
    identity mapping of $\cal H$, whose isotropy group in $U(2)\times U(2)$
    is $\left\{(g,h)\,|\,gh=1\right\}$. The local unitary
    automorphism group therefore is the co-diagonal subgroup
    \begin{displaymath}
      A(L,\sigma)
        =\left\{(g,h)\in SO(3)\by SO(3)\,\big|\,gh=1\right\},
    \end{displaymath}
    itself isomorphic to $SO(3)$.
\end{enumerate}
Returning to the general case the following observation is basic.
%
\begin{la}\label{projection_to_subsets}
Let $L\subset V^P$ be isotropic, $\sigma$ a section over $L$, and let $S\subset P$ be a subset. Then the Cartesian projection
\begin{displaymath}
  SO(3)^P\owns{(g_s)}_{s\in P}\mapstoo{(g_s)}_{s\in S}\in SO(3)^S
\end{displaymath}
sends $A(L,\sigma)$ into $A(L_S,\sigma|L_S)$.
\end{la}
\begin{proof}
The inclusion $C(L,\sigma)\hookrightarrow{\cal H}^{\otimes P}={\cal H}^{\otimes P\setminus S}\otimes{\cal H}^{\otimes S}$
becomes, under the identification
\begin{displaymath}
  \Hom\bigl(C(L,\sigma),
           {\cal H}^{\otimes P\setminus S}\otimes{\cal H}^{\otimes S}\bigr)
  =\Hom\bigl(C(L,\sigma)\otimes{\cal H}^{\otimes P\setminus S},
           {\cal H}^{\otimes S}\bigr)
\end{displaymath}
an operator whose image is just $C(L_S,\sigma|L_S)$.
\end{proof}
In the extreme case of a one element subset $S=\{p\}$ we write $A(L,\sigma)_p$ for the projection of $A(L,\sigma)$ in $SO(3)^{\{p\}}=SO(3)$ and call it the automorphism group \textit{at} the qubit $p$. It turns out that this group underlies strong restrictions whose nature depends on the position of $p$ with respect to the supports of elements of $L$\ts: the \textit{support} of a vector $v\in V^P$ is defined as
\begin{displaymath}
  \supp v=\left\{s\in P\,\big|\,v_s\neq0\right\}.
\end{displaymath}

One possible such restriction is given by the octahedral group generated by the rotations around the $x$, $y$, and $z$ axes by a right angle\ts: see Fig.\ts\ref{octa}. This subgroup of $SO(3)$ acts as the symmetric group on the set of the octahedron's surface normals through the origin, and in fact coincides with the Clifford group of a one qubit system.
\begin{figure}
  \begin{minipage}[t]{150pt}
  \psfrag{x}[ct][ct]{\footnotesize$x$}
  \psfrag{y}[lc][lc]{\footnotesize$y$}
  \psfrag{z}[cb][cb]{\footnotesize$z$}
  \includegraphics{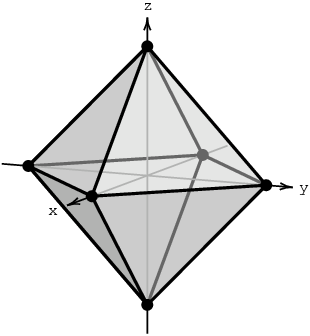}
    \caption{\label{octa} The standard octahedron}
  \end{minipage}
  \hfill
  \begin{minipage}[t]{150pt}
  \psfrag{x}[ct][ct]{\footnotesize$x$}
  \psfrag{y}[lc][lc]{\footnotesize$y$}
  \psfrag{z}[cb][cb]{\footnotesize$z$}
  \includegraphics{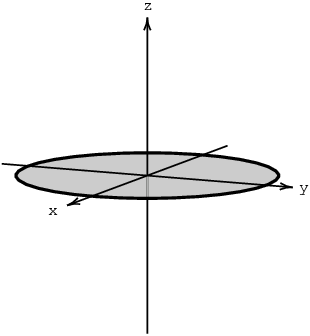}
    \caption{\label{dihedral} $D_z$ is the symmetry group of the standard disk}
  \end{minipage}
\end{figure}
\vspace*{12pt}
\begin{prop}{\bf(Rains' Lemma)}\label{rains}
  Assume that $|P|\ge3$ and let $L\subset V^P$ be an isotropic plane such
  that for every $p\in P$ the Cartesian projection $V^P\too V^{\{p\}}$ sends
  $L$ onto $V^{\{p\}}$ . Then for every $p\in P$ the group
  $A(L)_p\subset SO(3)$ is contained in the octahedral group.
\end{prop}
\begin{proof}
Pick any base $(v,w)$ of $L$ as a vector space\ts; the assumption on $L$ implies that both $v$ and $w$ have full support $P$. By a local symplectic automorphism and the standard choice of the section $\sigma$ we can normalise
\begin{displaymath}
  v=\sum_{s\in P}z_s\hbox{\q and\q}w=\sum_{s\in P}x_s,
\end{displaymath}
as well as
\begin{displaymath}
  \tilde L=\Bigl\langle\bigotimes_{s\in P}Z_s,
                       \bigotimes_{s\in P}X_s\Bigr\rangle.
\end{displaymath}
An element $g=\bigotimes_{s\in P}g_s\in LU(P)$ belongs to $A(L,\sigma)$ if and only if it commutes with the orthogonal projector to $C(L,\sigma)$
\begin{displaymath}
  \pi(L,\sigma)=\textstyle{1\over4}
    \left(1^{\otimes P}+Z^{\otimes P}+X^{\otimes P}
                +i^{|P|}Y^{\otimes P}\right)\in\End{\cal H}^{\otimes P},
\end{displaymath}
or equivalently with the traceless Hermitian endomorphism
\begin{displaymath}
  Z^{\otimes P}+X^{\otimes P}
        +i^{|P|}Y^{\otimes P}\in{(\Herm^0{\cal H})}^{\otimes P}
\end{displaymath}
--- note that $\omega(v,w)=0$ implies that $|P|$ must be even. Given $p\in P$ we put $S=P\sm\{p\}$ we take the partial dual of this endomorphism and obtain a $g$-equivariant linear operator
\begin{displaymath}
  \Herm^0{\cal H}\too{(\Herm^0{\cal H})}^{\otimes S}.
\end{displaymath}
This operator takes $Z$, $X$, and $Y$ to $Z^{\otimes S}$, $X^{\otimes S}$, and $i^{|P|}Y^{\otimes S}$ respectively. In view of $|S|\ge2$ each of these three tensors is decomposable, but this is true of no other linear combination of them. Therefore conjugation by $g$ must preserve the subset $\bb R\!\cdot\!Z\cup\bb R\!\cdot\!X\cup\bb R\!\cdot\!Y\subset\Herm^0{\cal H}$, and this means that the projection $g_p\in SO(3)$ belongs to the octahedral group.
\end{proof}
\vspace*{12pt}
Another kind of restriction may arise from one of the three infinite dihedral groups $D_a\subset SO(3)$ generated by all rotations around the $a$ axis and one rotation by 180 degrees around an orthogonal one, for $a\in\{x,y,z\}$\ts: see Fig.\ts\ref{dihedral}. A comprehensive listing of restrictions on the automorphism groups at the various qubits is provided by
\vspace*{12pt}
\begin{thm}\label{pointwise_groups}
Let $L\subset V^P$ be isotropic subspace that does not split off any zero factor nor a Bell Lagrangian. Then for each $p\in P$ at least one of the following statements is true.
\begin{enumerate}
  \item $A(L)_p$ is contained in the octahedral group, or
  \item $A(L)_p$ is contained in one of the dihedral groups $D_x$, $D_y$,
  $D_z$.
\end{enumerate}
\end{thm}
\begin{proof}
Given $p$ there exist vectors $v\in L$ with $v_p\neq0$, and we pick some such $v$ whose support is minimal among all supports containing $p$. By Lemma \ref{projection_to_subsets} we may assume that this support is the full set $P$. A local symplectic automorphism takes $v$ to the normal form $v=\sum_sz_s$, and we distinguish two cases\ts:
\begin{enumerate}
  \item If the Cartesian projection $V^P\too V^{\{p\}}$ sends $L$ onto
  $V^{\{p\}}$ then we pick another vector $w\in L$ such that $v_p\neq w_p$.
  In view of the minimality we must have $v_s\neq w_s$ for all
  $s\in P$, and may normalise to $w=\sum_sx_s$.

  \smallskip
  We claim that $L$ is the plane spanned by $v$ and $w$. Indeed, since
  clearly $|P|>1$ we may pick some $s\in P\sm\{p\}$\ts: then every vector
  $u\in L$ is congruent modulo $\langle v,w\rangle$ to another one, say
  $u'\in L$, with $u_s'=0$, and by minimality we conclude $u'=0$.
  
  Since $L$ does not contain a Bell factor we even have $|P|\ge3$, and
  Rains' Lemma now implies that $A(L)_p$ is contained in the octahedral
  group.
\medskip
  \item In case the projection of $L$ to $V^{\{p\}}$ is one-dimensional we
  put $S=P\sm\{p\}$\ts; for suitable (indeed half the possible) choices of
  $\sigma$ the projector to $C(L,\sigma)$ is
  \begin{displaymath}
    \pi(L,\sigma)=\textstyle{1\over2}\left(1+Z^{\otimes P}\right)\circ
                           \bigl(1\otimes\pi(L_S,\sigma|L_S)\bigr).
  \end{displaymath}
  Thus for every $g={(g_s)}_{s\in P}\in A(L,\sigma)$ the plane
  $\langle1,Z\rangle\subset\Herm{\cal H}$ must be stable under conjugation by
  $g_p$, and we conclude $g_p\in D_z$.
\end{enumerate}
\end{proof}
\vspace*{12pt}
Recall that a group is called solvable if it can be represented as a successive extension of abelian groups. Examples include the octahedral and dihedral groups, but not, for instance, the symmetry group of an icosahedron.
\vspace*{12pt}
\begin{cor}\label{solvable}
Let $L\subset V^P$ be isotropic subspace that does not split off any zero factor nor a Bell Lagrangian. Then $A(L)$ is a solvable compact Lie group.
\end{cor}
\begin{proof}
It suffices to observe that by Theorem \ref{pointwise_groups}
\begin{displaymath}
  A(L)\subset\prod_{p\in P}A(L)_p
\end{displaymath}
is a closed subgroup of a Cartesian product of solvable compact Lie groups.\end{proof}
\vspace*{12pt}
Given any $L$ and $\sigma$ with $\dim L=|P|\!-\!d$ let
\begin{displaymath}
  Z(L,\sigma)=\left\{g\in A(L,\sigma)\,\big|\,\tilde g|C(L,\sigma)
                 \hbox{ is scalar for some }\tilde g
                 \hbox{ representing }g\right\}
\end{displaymath}
be the normal subgroup of automorphisms which act by scalars on the code space. Up to scalar factors the quotient group $\overline A(L,\sigma):=A(L,\sigma)/Z(L,\sigma)$ is the group of unitary operators on $C(L,\sigma)$ that can be (exactly) realised \textit{transversally} on ${\cal H}^{\otimes P}$. If $\tilde L^\perp\subset\tilde V^P$ denotes the inverse image of $L^\perp$ under the projection $\tilde V^P\to V^P$
then an embedding of groups
\begin{displaymath}
  L^\perp/L\simeq\tilde L^\perp/(S^1\tilde L)\too\overline A(L,\sigma)
\end{displaymath}
is induced, and the identification of $C(L,\sigma)$ with a space ${\cal H}^{\otimes d}$ of logical qubits can be ---and usually is --- made in such a way that $\tilde L^\perp/\tilde L$ acts as the logical Pauli group. In the present context, though, any identification will do.
\vspace*{12pt}
\begin{cor}\label{no_transv}
Let $L\subset V^P$ be an isotropic subspace that does not split off a zero factor. Then $\overline A(L)$ is a solvable compact Lie group. In particular no universal set of operators even on one logical qubit can be realised transversally.
\end{cor}
\begin{proof}
For a Bell Lagrangian $L$ we have $Z(L)=A(L)\simeq SO(3)$, so that $\overline A(L)$ is trivial. We thus may assume that $L$ does not contain any such factor. Now given an arbitrary isomorphism $C(L)\simeq{\cal H}^{\otimes d}$ the elements of $A(L)$ that induce a one-qubit operator on a fixed logical qubit form a closed subgroup of $A(L)$. Its image in the copy of $SO(3)$ corresponding to that qubit is still solvable\ts: it must be a proper subgroup since $SO(3)$ is a simple group. As it is a closed subgroup the class of operators that allow but approximate transversal realisations is no larger.
\end{proof}
\vspace*{12pt}
Note that in view of the known classification of subgroups of $SO(3)$ the last conclusion of the corollary effectively restricts the set of transversal one qubit gates to a subgroup of either an octahedral or an infinite dihedral group.
\section{The connected group of automorphisms}
\label{connected_group}
\noindent
In this section we determine the connected component $A_c(L)$ of the identity operator in the automorphism group $A(L)$. As this group has already been determined for zero factors as well as trivial and Bell Lagrangians we assume that the isotropic subspace $L$ does not contain any of these. Given $L\subset V^P$ we call two qubits $p,q\in P$ \textit{contiguous} if
\begin{displaymath}
  p=q\hbox{, or }L\hbox{ contains a vector }v\hbox{ with }\supp v=\{p,q\}.
\end{displaymath}
Contiguity is an equivalence relation since for pair-wise distinct qubits $p,q,r\in P$ and vectors $v,w\in L$ with $\supp v=\{p,q\},\;\supp w=\{q,r\}$ the isotropy of $L$ enforces $v_q=w_q$ and thereby $\supp(v\!+\!w)=\{p,r\}$. We denote by $\cal S$ the set of contiguity classes.

Recall that $(z,x)$ denotes the standard base of the vector space $V=\bb F_2\oplus\bb F_2$. We append the subscript $p$ when referring to these vectors in the $p$-th copy $V^{\{p\}}$ of $V$. 
\vspace*{12pt}
\begin{la}\label{normalise}
Let $L\subset V^P$ be an isotropic subspace free of zero, trivial, and Bell factors. Applying a suitable local symplectic automorphism of $V^P$ we can then achieve that for each $S\in{\cal S}$ either
\begin{displaymath}
  L_S=\langle z_p\!+\!z_q\,|\,p\neq q\rangle
\end{displaymath}
or
\begin{displaymath}
  L_S=\langle z_p\!+\!z_q\,|\,p\neq q\rangle
     +\Bigl\langle\sum_{s\in S}x_s\Bigr\rangle.
\end{displaymath}
\end{la}
\begin{proof}
For $|S|=1$ the conclusion holds trivially, and we thus assume $|S|\ge2$. Fix any $q\in S$, and for each $p\in S\sm\{q\}$ pick a vector $v(p)\in L$ with $\supp v(p)=\{p,q\}$. Since $L$ is isotropic the component $v(p)_q\in V$ is independent of $p$, and by a local symplectic transformation we achieve that $v(p)=z_p\!+\!z_q$ holds for all $p\neq q$. Thus $L$ contains the subspace $\langle z_p\!+\!z_q\,|\,p\neq q\rangle$ of dimension $|S|\!-\!1$.

In case this is a proper subspace of $L_S$ we choose an arbitrary $w=\sum_{s\in S}w_s$ in the difference $L\sm\langle z_p\!+\!z_q\,|\,p\neq q\rangle$. Since $L$ has no trivial factor at least one $w_s$ differs from $z_s$, and as $L$ is isotropic this is even true for every $s\in S$. A further local symplectic transformation fixing the $z_s$ now moves $w_s$ to $x_s$ for each $s\in S$.
\end{proof}
\vspace*{12pt}
Note that in the second case of Lemma \ref{normalise} the isotropic space $L_S$ is a Lagrangian defining a Greenberger-Horne-Zeilinger (GHZ) state and necessarily a direct factor of $L$.

Let $L$ be given as before. Following \cite{poulin} we call an element $p\in P$ a \textit{protected} qubit if the Cartesian projection
\begin{displaymath}
  L\hookrightarrow V^P\too V^{\{p\}}
\end{displaymath}
is surjective, while otherwise --- when $L$ projects to a line --- we call $p$ a \textit{gauge} qubit. By Lemma \ref{normalise} these properties only depend on the contiguity class of $p$, so that we may speak of classes $S\in{\cal S}$ that are protected respectively of gauge type.

We can now state our result. Recall that $C_a\subset SO(3)$ is the circle group of rotations around the $a$ axis.
\vspace*{12pt}
\begin{thm}\label{connected}
Let $L\subset V^P$ be an isotropic subspace free of zero, trivial, and Bell factors. Then after a suitable local symplectic transformation and for a suitable choice of the section $\sigma\from L\to\tilde L$ the connected automorphism group becomes
\begin{displaymath}
  A_c(L,\sigma)=\Bigl\{{(g_p)}_{p\in P}\in {(C_z)}^P\;\Big|\;
                  \prod_{s\in S}g_s=1
                  \hbox{ for every protected class }S\in{\cal S}\Bigr\}.
\end{displaymath}
Thus the dimension of this torus is $|P|$ minus the number of protected contiguity classes.
\end{thm}
\begin{proof}
We proceed in four steps. We first describe the necessary symplectic transformations, then we show in Lemma \ref{z-axis} that all automorphisms are composed of rotations about the various $z$ axes. Lemma \ref{is_large} then shows that $A_c(L,\sigma)$ is at least as large as stated, and finally in Lemma \ref{is_small} we prove that it is no larger.

We begin by normalising $L$ as in Lemma \ref{normalise}. We are then still free to apply an arbitrary symplectic automorphism to each factor $V^S$ with $S\in{\cal S}$ and $|S|=1$. If such an $S$ is of gauge type we thus achieve that $L$ projects onto the line $\langle z\rangle\subset V^S$. On the other hand if $S=\{p\}$ is protected we make use of the fact that the connected part of $A(L)_p$ is contained in one of the circle groups $C_x$, $C_y$, $C_z$, and apply a local symplectic transformation to obtain that $A(L)_p\subset C_z$.

Another adjustment is made to those factors $V^S$ such that $S\in{\cal S}$ is a protected class with $|S|>1$ and $L_S=\langle z_p\!+\!z_q\,|\,p\neq q\rangle$\ts: the projection of $L$ into $V^S$ must contain a vector $w$ with $w_s\neq z_s$ for at least one and therefore every $s\in S$, and the local symplectic automorphism fixing $z_s$ and sending $w_s$ to $x_s$ transforms $w$ into $\sum_{s\in S}x_s$ without disturbing $L_S$.

We also normalise the choice of the section $\sigma\from L\to\tilde L$, requiring that $\tilde L$ contains the operators $Z_pZ_q\in\tilde V^P$ whenever the qubits $p$ and $q$ are contiguous.
\vspace*{12pt}
\begin{la}\label{z-axis}
If $L$ and $\sigma$ are normalised in this way then for every party $p\in P$ the connected part of $A(L,\sigma)_p$ is contained in $C_z$.
\end{la}
\begin{proof}
There is nothing to show if $p$ is protected and forms a class by itself. In case $p$ is of gauge type and $\{p\}\in{\cal S}$ we write
\begin{displaymath}
  V^P=V^{\{p\}}\times V^T\hbox{\q with }T:=P\sm\{p\}
\end{displaymath}
and pick a $v\in L$ with $v_p=z$. The Pauli operator $\sigma(v)\in\Herm{\cal H}\otimes\Herm{\cal H}^{\otimes T}$ has the form $Z\otimes\tilde v$, and the projector to the code space $C(L,\sigma)$ becomes
\begin{displaymath}
  \pi(C,\sigma)=\textstyle{1\over2}\bigl(1+Z\otimes\tilde v\bigr)
                    \circ\bigl(1\otimes\pi(L_T,\sigma|L_T)\bigr).
\end{displaymath}
Rewriting this as a linear operator $\Herm{\cal H}\to\Herm{\cal H}^{\otimes T}$ we conclude that $A(C,\sigma)_p$ is contained in the dihedral group $D_z$.

Finally we consider the case that $p$ belongs to a class $S\in{\cal S}$ with $|S|>1$. Since $L_S$ must be one of the two possibilities stated in Lemma \ref{normalise} and since by Lemma \ref{projection_to_subsets} we may assume $S=P$, there are but two explicit cases for $L$ to verify. The standard choice of the section $\sigma\from L\to\tilde L$ leads to the code space 
\begin{displaymath}
  \bigl\langle|00\,\cdots\,0\rangle,|11\,\cdots\,1\rangle\bigr\rangle
    \subset{\cal H}^{\otimes P}
\end{displaymath}
respectively the GHZ state
\begin{displaymath}
  |00\,\cdots\,0\rangle+|11\,\cdots\,1\rangle\in{\cal H}^{\otimes P}
  \hbox{\q with }|P|\ge3.
\end{displaymath}
In each case even the full automorphism group is easily read off\ts: it comprises the simultaneous rotation by 180 degrees around each $x$ axis as well as free rotations $g_s\in C_z$ in each factor $s\in S$, restricted by the condition $\prod_{s\in P}g_s=1$ in the GHZ case.
\end{proof}
%
\vspace*{12pt}
\begin{la}\label{is_large}
Let $L$ and $\sigma$ be normalised. Then for every $S\in{\cal S}$ the torus \begin{displaymath}
  \Bigl\{{(g_p)}_{p\in S}\in {(C_z)}^S\;\Big|\;\prod_{s\in S}g_s=1
    \hbox{ if }S\hbox{ is protected }\Bigr\}
\end{displaymath}
is contained in $A(L,\sigma)$.
\end{la}
\begin{proof}
For all $p,q\in S$ the stabilizer group $\tilde L=\sigma(L)$ contains the Pauli product $Z_pZ_q$. Therefore if a state
\begin{displaymath}
  |\psi\rangle=\sum_{u\in\bb F_2^P}\lambda_u\cdot|u\rangle\in C(L,\sigma)
\end{displaymath}
is written out in the computational basis of ${\cal H}^{\otimes P}$, the coefficient $\lambda_u$ must vanish unless all $u_s\in\bb F_2$ with $s\in S$ are equal to each other. Thus every product ${(g_p)}_{p\in S}\in {(C_z)}^S$ with $\prod_{s\in S}g_s=1$ leaves $C(L,\sigma)$ point-wise fixed and in particular belongs to $A(L,\sigma)$.

Nothing more is claimed if the contiguity class $S$ is protected. In case it has gauge type let us put $T:=P\sm S$ and split $V^P=V^S\times V^T$. Since $L$ is normalised the projection $V^P\to V^S$ sends $L$ into $\langle z_s\,|\,s\in S\rangle$, and we find a vector $v\in L$ such that
\begin{displaymath}
  L=L_S+\langle v\rangle+L_T
\end{displaymath}
and $v_s\in\langle z_s\rangle$ for all $s\in S$. We write $\sigma(v)=\tilde v_S+\tilde v_T\in\tilde V^S\oplus\tilde V^T$ and note that $\tilde v_S$ is a product of Pauli $Z$ operators\ts; thus every ${(g_p)}_{p\in S}\in {(C_z)}^S$ indeed commutes with the projector
\begin{displaymath}
  \pi(L,\sigma)
    =\textstyle{1\over2}\bigl(1+\tilde v_S\otimes\tilde v_T\bigr)
     \circ\bigl(\pi(L_S,\sigma)\otimes\pi(L_T,\sigma)\bigr).
\end{displaymath}
\end{proof}
\vspace*{12pt}
By the \textit{support} of an element $g={(g_p)}_{p\in P}\in LU(P)$ we mean, of course, the set of $p\in P$ with $g_p\neq1$. For a given normalised isotropic $L\subset V^P$ we now fix a subset $T\subset P$ that contains exactly one party from each protected class $S\in{\cal S}$, and does not meet any contiguity class of gauge type.
\vspace*{12pt}
\begin{la}\label{is_small}
Every one parameter group $S^1\to A(L,\sigma)$ of automorphisms with support in $T$ is trivial.
\end{la}
\begin{proof}
For fixed $p\in T$ we consider the infinitesimal rotation $Z_p$ in the factor $p$\ts: on the projector $\pi(L,\sigma)$ it acts  by
\begin{displaymath}
  [Z_p,\pi(L,\sigma)]
    =\left[Z_p,\textstyle{1\over|L|}\sum_{v\in L}\sigma(v)\right] 
    =\textstyle{2\over|L|}\!\!\!\!
      \disp\sum_{v\in L\atop v_p\in\{x,z+x\}}\!\!\!\!\!\!
      Z_p\,\sigma(v).
\end{displaymath}
For a general infinitesimal one parameter group of $A(L,\sigma)$ with support in $T$, say $\sum_{p\in T}\gamma_pZ_p$, we thus obtain
\begin{displaymath}
  0=\Bigl[\sum_{p\in T}\gamma_pZ_p,\pi(L,\sigma)\Bigr]
   =\textstyle{2\over|L|}\;
     \disp\sum_{p\in T}\;
     \gamma_p\cdot\!\!\!\!\!\!
     \disp\sum_{v\in L\atop v_p\in\{x,z+x\}}\!\!\!\!\!\!
     Z_p\,\sigma(v).
\end{displaymath}
On the right hand side each term of the inner sum is a tensor product of Pauli operators in ${(\Herm{\cal H})}^{\otimes P}$. In view of $z_p+z_q\notin L$ for every two distinct $p,q\in T$ these operators are pair-wise distinct even up to scalar factors, and therefore form a linearly independent system.

On the other hand each $p\in T$ is protected, so that the projection of $L$ to $V^{\{p\}}$ is surjective. In particular none of the inner sums is empty\ts: this implies $\gamma_p=0$ for all $p\in T$.
\end{proof}
\vspace*{12pt}
We now readily complete the proof of Theorem \ref{pointwise_groups}\ts: using Lemma \ref{is_large} every one parameter group of $A(L,\sigma)$ may be reduced to one with support in $T$, which by Lemma \ref{is_small} must be trivial.
\end{proof}
\vspace*{12pt}
As noted in the proof of Lemma \ref{is_large}, for each class $S\in{\cal S}$ the automorphisms ${(g_p)}_{p\in S}\in {(C_z)}^S$ with $\prod_{s\in S}g_s=1$ act trivially on the code space $C(L,\sigma)$. We thus conclude\ts:
\vspace*{12pt}
\begin{cor}\label{finite}
Let $L\subset V^P$ be an isotropic subspace without zero factors. If every nontrivial qubit of $P$ is protected then the group $\overline A(L)=A(L)/Z(L)$ is finite.
\end{cor}
\vspace*{12pt}
\section{Conclusion and examples}
\label{concl}
\noindent
We have shown that the automorphism group of a stabilizer code usually belongs to the class of solvable groups, and have thus given a precise and conceptual meaning to the notion that there are but few automorphisms. This in particular re-proves the theorem on the failure of transversal gates to form a universal set of encoded gates.

We have further computed the connected automorphism group of a stabilizer code or state. For the purpose of coding these automorphisms turn out to be of little help\ts: the transversal gates they induce can only act on gauge qubits, whose presence of course precludes any error correcting capacity of the code.

By contrast there are quite a few pure stabilizer states with non-trivial connected automorphism group. Apart from GHZ states in any number $|P|>0$ of qubits the simplest examples arise from a partition $P=S+T$ of the set of qubits with $|S|\ge2\le|T|$ and the Lagrangian
\begin{displaymath}
  \bigl\langle z_s\!+\!z_p\,|\,s\in S\bigr\rangle
   +\Bigl\langle\sum_{s\in S}x_s\!+\!z_q,\,
                z_p\!+\!\sum_{t\in T}x_t\Bigr\rangle
   + \bigl\langle z_t\!+\!z_q\,|\,t\in T\bigr\rangle,
\end{displaymath}
with arbitrarily chosen qubits $p\in S$ and $q\in T$ (all choices giving the same Lagrangian of course). Indeed the symmetry of the corresponding state
\begin{displaymath}
  |0\rangle\!\otimes\!|0\rangle+|0\rangle\!\otimes\!|1\rangle
 +|1\rangle\!\otimes\!|0\rangle-|1\rangle\!\otimes\!|1\rangle
 \in{\cal H}^{\otimes S}\otimes{\cal H}^{\otimes T}
\end{displaymath}
is clearly visible.

As was to be expected the existence of a non-trivial connected group of automorphisms of a stabilizer state is the exception rather than the rule\ts: most often the automorphism group is finite. An investigation of this finite group will be much more delicate, and quite different phenomena have to be taken into account, as is suggested by the following well-known example derived from Reed-Muller codes.

The shortened Reed-Muller code $R_\ast(r,m)$ is a classical binary code on the set $P=\bb F_2^m\sm\{0\}$, comprising all $v={(v_s)}_{s\in P}$ with the property that the function
\begin{displaymath}
  P\owns s\mapsto v_s\in\bb F_2
\end{displaymath}
is polynomial of degree at most $r$ and without constant term. The two codes $R_\ast(m\!-\!2,m)$ and $R_\ast(1,m)$ are orthogonal to each other, so that using the former in the $z$- and the latter in the $x$-co-ordinate space we obtain an isotropic subspace
\begin{displaymath}
  L:=R_\ast(1,m)\oplus R_\ast(m\!-\!2,m)\subset\bb F_2^P\oplus\bb F_2^P=V^P.
\end{displaymath}
The corresponding stabilizer code is an indecomposable Calderbank-Shor-Steane code which encodes just one qubit. For the standard choice of $\sigma\from L\to\tilde L$ the simultaneous rotation by $360/2^{m-1}$ degrees around each $z$ axis yields an automorphism $g\in A(L,\sigma)$ of order $2^{m-1}$, and for $m\ge4$ this $g$ is unique of this order up to multiplication by automorphisms of lower order (dividing $2^{m-1}$). This is in striking contrast to the presence of automorphisms with small support in the case of a non-trivial connected automorphism group. 
%
%
%
\bibliography{automorphisms}
\end{document}